\pdfoutput = 1

\documentclass[%
 reprint,
 amsmath,amssymb,
 aps,longbibliography,
]{revtex4-2}

\usepackage{graphicx}
\usepackage{dcolumn}
\usepackage{bm}

\usepackage{hyperref}

\usepackage{booktabs}
\usepackage{tabularx}
\usepackage{amsmath,amsfonts,amssymb,amsthm}
\usepackage[caption=false]{subfig}

\renewcommand{\d}{\mathrm{d}}

\newcommand{\so}{M_{\text{S}}}
\newcommand{\io}{M_{\text{I}}}
\newcommand{\aw}{M_{\text{SI}}}
\newtheorem{remark}{Remark}
\newtheorem{lemma}{Lemma}
\newtheorem{theorem}{Theorem}

\usepackage{comment}
\usepackage{appendix}
\usepackage{chngcntr}

\begin{document}


\title{Epidemic paradox induced by awareness driven network dynamics}
\author{
\normalsize{
Csegő Balázs Kolok, \textsuperscript{1,2}
Gergely \'Odor,\textsuperscript{3,2}
Dániel Keliger,\textsuperscript{4,2}
M\'arton Karsai,\textsuperscript{3,2}
}
\\
\small{
1. Institute of Mathematics, Faculty of Science, ELTE Eötvös Loránd University, Budapest, Hungary\\
2. National Laboratory of Health Security, HUN-REN Alfr\'ed R\'enyi Institute of Mathematics, Budapest, Hungary\\
3. Department of Network and Data Science, Central European University, Vienna, Austria\\
4. Department of Stochastics, Institute of Mathematics, \\ Budapest University of Technology and Economics, Budapest, Hungary
}
}

\date{\today}
\begin{abstract}

We study stationary epidemic processes in scale-free networks with local awareness behavior adopted by only susceptible, only infected, or all nodes. We find that while the epidemic size in the susceptible-aware and the all-aware models scales linearly with the network size, the scaling becomes sublinear in the infected-aware model. Hence, fewer aware nodes may reduce the epidemic size more effectively; a phenomenon reminiscent of Braess’s paradox. We present numerical and theoretical analysis, and highlight the role of influential nodes and their disassortativity to raise epidemic awareness.

\end{abstract}


\maketitle

\begin{figure}[b]
\includegraphics[width=0.5\textwidth]{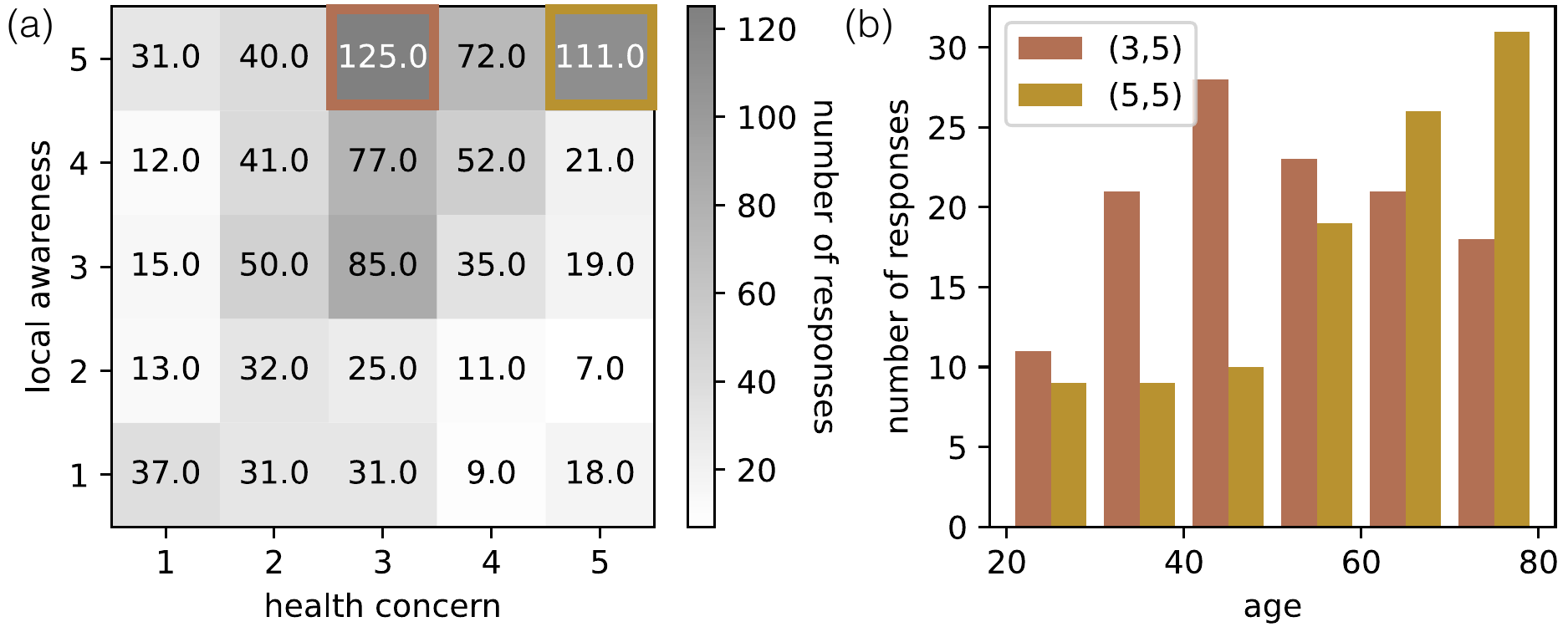} 
\caption{(a) 2D histogram of local awareness and concerns of respondents about their own health in case of an infection in the MASZK survey. The two most frequent answers are framed and written in white. (b)  The age distribution of the two most frequent answers from subfigure (a). Among the most aware respondents, those who had medium health concerns (3) were middle-aged (red) and those who had the highest health concerns (5) were older (yellow). }
\label{fig:1}
\end{figure}

Epidemics drive substantial changes in human behavior, affecting our everyday lives, social interactions, and economic activities \cite{champion2008health, abraham2015health}. Most recently, amid the COVID-19 pandemic, people have heightened their protective measures \cite{JOSE202141, BORKOWSKI2021102906}, including mask wearing~\cite{10.1371/journal.pone.0240785}, increased hygiene practices \cite{gibson2020capability}, social distancing \cite{NBERw28139, andersen2020early}, and the avoidance of large gatherings \cite{ebrahim2020covid, PhysRevLett_social}. These behavioral changes are intended to reduce disease transmission, which is often captured in behavior-disease -- also called ``awareness'' -- models \cite{funk2010modelling,wang2015coupled,verelst2016behavioural,teslya2020impact}. 

In this Letter, we focus on a subtle, but powerful form of awareness, called ``local awareness'' \cite{wu2012impact}, which describes the behavioral changes adopted by individuals when they learn about an increase in the prevalence of the disease within their direct contacts, like among their family members or friends. In such situations, individuals may take preventive measures more seriously causing a more substantial reduction in the disease transmission. Several models have been proposed to quantify the impact of local awareness in epidemic models \cite{PhysRevLett_adaptive, funk2009spread, kiss2010impact, perra2011towards}, some of which were validated in real epidemic scenarios~\cite{odor2024epidemic}. Our goal is to demonstrate how small changes in the awareness model -- inspired by recent survey results -- can cause substantial and even paradoxical changes in the resulting epidemic dynamics and outcome.

Our research design was inspired by the Hungarian MASZK COVID-19 questionnaire \cite{maszk}, where we asked respondents to rate their willingness to engage in local awareness, and about their concern regarding their own health in case of an infection on a scale of 1 to 5. Fig. \ref{fig:1}(a) shows that the respondents clustered into two distinct groups: those who rated both their health concerns and their awareness level high (5,5), and those who rated their health concerns medium and their local awareness level high (3,5). The former group contains predominantly middle-aged individuals, while the latter group contains predominantly older ones (Fig. \ref{fig:1}(b)). This observation is consistent with previous COVID-19 studies revealing different motivations to engage in protective behavior. In these survey experiments, younger people were found to be more concerned about infecting others while elder people tended to focus on protecting themselves \cite{shiina2020relationship,germani2020emerging}. It has also been shown that age is not the sole determining factor~\cite{pappa2021recession}; occupation and socioeconomic status may also influence motivations for protective behaviors.

To investigate the impact of motivational differences to engage in awareness behavior, we extend the Susceptible-Infected-Susceptible (SIS) epidemic model with three distinct awareness models: (i) S-aware, where only the susceptible individuals are aware to protect their own health; (ii) I-aware, where only the infected individuals are aware to protect their peers; and (iii) SI-aware, where both susceptible and infected individuals engage in awareness behavior. While the effects of these awareness models have not been fully explored at the agent level, previous studies have shown that local awareness, in general, can raise the epidemic threshold \cite{wu2012impact,sahneh2012existence,liu2015interplay,moinet2018effect}, and I-awareness can reduce the epidemic size more effectively than S-awareness~\cite{funk2010modelling}.

Going beyond these initial observations, our study aims to understand how awareness impacts the epidemic spreading dynamics in the three defined models. Our simulations reveal a counter-intuitive result, showing that higher awareness levels can lead to larger epidemics, which is reminiscent of Braess's paradox \cite{braess1968paradoxon}. The connection between epidemics and Braess's paradox has previously been made by \cite{zhang2013braess} in a game theoretic setting. In contrast, our result is an emergent network phenomenon, which we support by a theoretical analysis, a demonstration on real social networks, and an intuitive explanation of the metastable state of the SIS model in the three awareness scenarios.

\textit{Models and methods.} 
Networks are defined by a set of nodes $V$ and edges $E$, which represent connections between the nodes. We focus on random networks with a power-law degree distribution generated by the Chung-Lu model \cite{chunglu, fasino2021generating}, a generalization of the Erdős-Rényi model. In a network of size $n$, we add an edge between nodes $i$ and $j$ with probability $ w_{ij}:=\kappa \frac{d_i d_j}{D}$,  where $D=\sum_i d_i$ is a normalization term, $\kappa>0$ is a density parameter, and $d_i$ is a power-law sequence with exponent~$\gamma$~\cite{fasino2021generating}. Note that vertex $i$ has expected degree $\kappa d_i$, and that the second moment of the degree sequence is infinite for $\gamma \in (2,3)$, and finite if $\gamma > 3$.

An SIS process assumes each network node to be either susceptible or infected. We initialize all nodes to be susceptible, except a small, randomly selected set of infected seed nodes. Thereafter, in each step, an infected node infects its susceptible neighbors with probability $\beta_0$, while recovering with probability~$\mu$. Eventually, the SIS process is known to reach a metastable state, where the number of infected individuals does not change beyond statistical fluctuations. We denote this metastable epidemic size by $I_{\infty}$. 

To introduce local awareness to the SIS model we use an exponentially decaying awareness function inspired by the models of \cite{zhang2014suppression,wu2012impact}. If an infected node $u$ and a susceptible node $v$ are connected by an edge, $u$ infects $v$ with probability
\begin{equation}
\label{eq:betadef}
    \beta(u,v,t) = \beta_0 a(v,t)^{\alpha_S} a(u,t)^{\alpha_I},
\end{equation}
where $a(v,t) =  e^{- N_I(v,t)}$ is the awareness function of an aware node $v$ (function $a(u,t)$ has the same form), $N_I(v,t)$ denotes the number of infected neighbors of node $v$ at time $t$, and the constants $\alpha_S$ and $\alpha_I$ determine the nature of the local awareness model in the population. We discuss why the awareness function depends on the number and not the proportion of the infected neighbors in Supplementary Section 1. We distinguish three notable prevalence-based local awareness models:
\begin{itemize}
    \item[(i)] S-aware ($\so$), where $\alpha_S=1$ and $\alpha_I=0$ (only susceptible nodes are aware) , 
    \item[(ii)] I-aware ($\io$), where  $\alpha_S=0$ and $\alpha_I=1$ (only infected nodes are aware),
    \item[(iii)] SI-aware ($\aw$), where $\alpha_S=\alpha_I=1$ (infected nodes in the neighborhoods of both susceptible and infected nodes are aware).
\end{itemize}
Fig. \ref{fig:2}(a) illustrates how the infected neighbors $N_I(u,t)$ and $N_I(v,t)$ are counted in the exponent of the infection probabilities of the three awareness models. 

\begin{figure}[b]
\includegraphics[width=1\columnwidth]{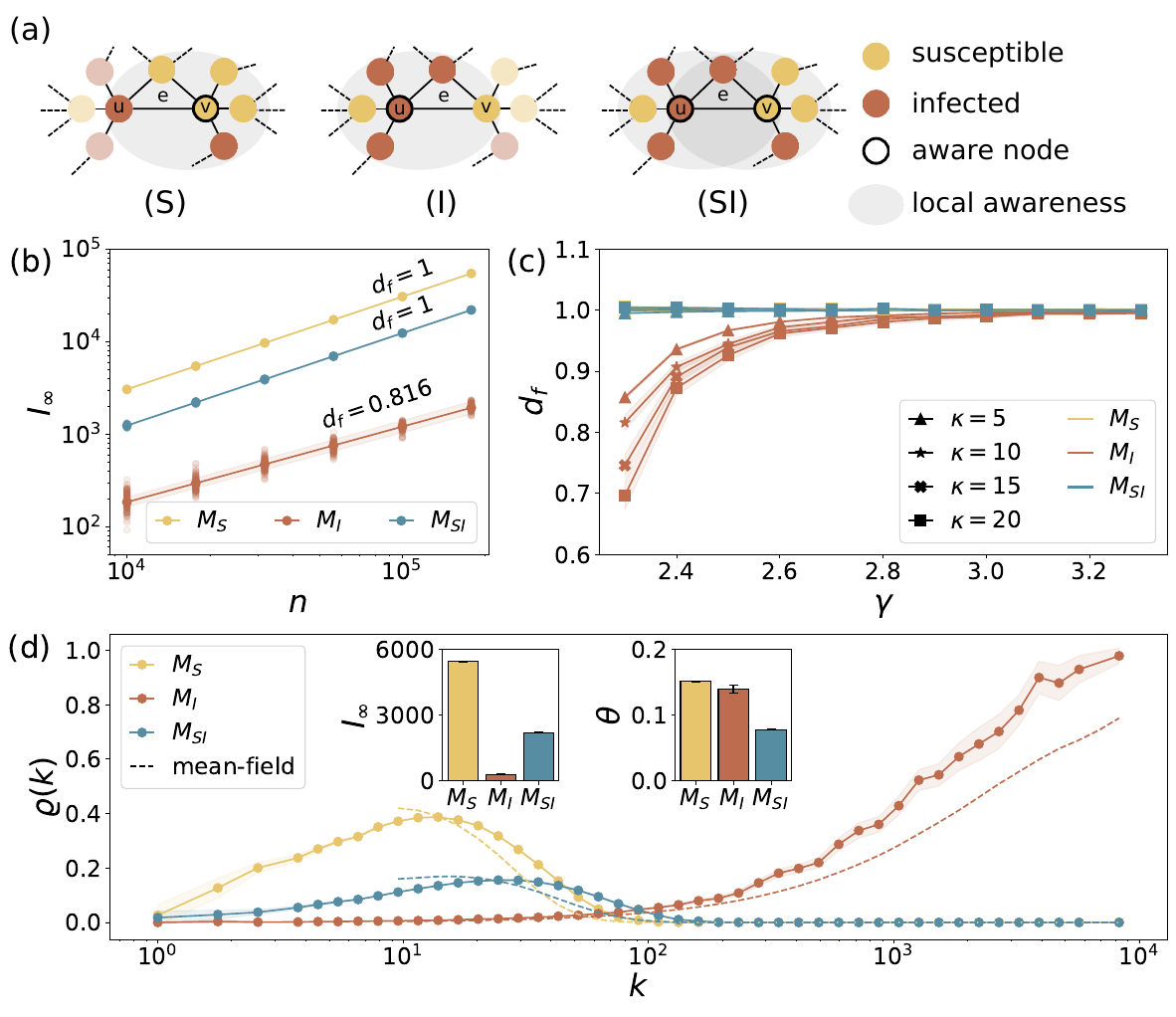}
\caption{(a) Illustration of the set of nodes (in shaded areas) counted in the exponents of the awareness functions in the (S)~S-aware $\so$, (I) I-aware $\io$ and (SI) SI-aware $\aw$ models. (b) The fractal dimension $d_f$ measured as the slope of the metastable epidemic size $I_{\infty}$ as a function of the network size $n$ in Chung-Lu networks with $\gamma = 2.3$ and $\kappa = 10$. (c) For degree exponent $\gamma < 3$, the fractal dimension of the $\io$ is smaller than 1, whereas the fractal dimensions of the $\so$ and the $\aw$ remain 1. Paradoxically, the infection becomes smaller in the $\io$ despite that more nodes are aware in the $\aw$. (d) The density of infection in the metastable state $\varrho(k)$ as function of the node degree $k$ both in simulations (solid) and in mean-field numerical approximations (dashed). Infection is concentrated on low-degree nodes in the $\so$ and $\aw$ models, while for $\io$, the high degree nodes dominate the infection. Inset (left) verifies that the metastable epidemic size is the lowest for $\io$. Inset (right) shows no paradox in the perceived infection density $\theta$, defined as the probability that a randomly chosen edge has an infected node on its end. 
}
\label{fig:2}
\end{figure}

\textit{Results.} Since more nodes are aware in the $\aw$  compared to $\so$ and $\io$, one would expect that $\aw$ would consistently lead to the lowest number of infected nodes. Interestingly, we find that this is not always the case.

To assess differences between the three models in a robust way, we measure the asymptotic growth of the metastable epidemic size $I_{\infty}$ as a function of the network size $n$. More precisely, we assume the polynomial relation $I_{\infty} \sim n^{d_f}$, and we are interested in the exponent $d_f$, also called the fractal dimension of the metastable state. Numerically, $d_f$ can be estimated by the slope of $I_{\infty}$ as a function of $n$ on a log-log plot (see Fig. \ref{fig:2}(b)).

In supercritical SIS models without awareness, the metastable infection size $I_{\infty}$ is known to be linear in the network size \cite{newman2018networks}, which means that the fractal dimension attains its maximum value ($d_f=1$). Fig. \ref{fig:2}(c) shows that the same is true for $\so$ and $\aw$, across all values of $\gamma$, suggesting that these awareness mechanisms are unable to asymptotically reduce the epidemic size. However, surprisingly, Fig. \ref{fig:2}(c) also shows that $\io$ has fractal dimension strictly less than 1, implying sublinear growth, for networks with $\gamma<3$. In this context, sublinearity means that for large enough networks, the density of infectious individuals in the stationary state can be lower than any fixed number, in particular, lower than the stationary infection densities in the $\so$ and $\aw$ models. This is a highly counter-intuitive result: even though less nodes are aware in $\io$, the epidemic becomes smaller compared to $\aw$, when all nodes reduce contacts (Fig.~\ref{fig:2}(d), left inset). Additionally, as the network density $\kappa$ increases, the fractal dimension further decreases in the $\io$ model, strengthening the paradoxical phenomenon. 

Fig. \ref{fig:2}(d) shows that the three awareness models lead to markedly different degree profiles among the infected nodes, hinting at an intuitive explanation of the paradox. In the $\so$ and $\aw$ models, the infection density is highest among low-degree nodes (degree between 10-50), and since they constitute the overwhelming majority of the network, the epidemic size is also high. In contrast, in the $\io$ model, the infection is concentrated on the high-degree nodes (degrees larger than 1000), and low-degree nodes remain sparsely infected, making the infection size significantly smaller. In other words, even though the $\aw$ model has the highest potential to reduce the epidemic size, due to the way the epidemic is distributed in the metastable state, awareness has a bigger impact in the $\io$ model. 


We observe no paradoxical behavior, if we evaluate the infection density ``perceived'' by the nodes, defined as the fraction of infected neighbors $\theta$ of a randomly chosen individual. Indeed, the second inset of Fig. \ref{fig:2}(d) shows that $\theta$ is smallest for the $\aw$ model. This observation suggests that in the $\io$ model, a few infected high-degree nodes increase the perceived epidemic density in their large neighborhoods, serving as ``warning examples'' and protecting the rest of the population from the disease. However, if these nodes are more interested in protecting themselves, as in the case of the $\aw$ and $\so$ models, then high-degree nodes do not become infected, and therefore do not alert their neighbors; instead they are essentially removed from the network. Although the removal of high-degree nodes does slow down the epidemic propagation compared to a null-model without awareness, the network between the low-degree nodes remains intact, and a constant fraction of nodes become infected in the stationary state. Eventually, the perceived epidemic density becomes comparable in all three behavior-epidemic models, however, this means a significantly higher number of infected low-degree nodes in the $\aw$ and $\so$ models, compared to the few infected hubs in the $\io$ model. 

We also observe no paradoxical behavior in the first few timesteps of the three awareness models, initialized from the same seed nodes. Instead, high degree nodes become infected first in all three models (Supplementary Section 6), and the epidemic grows slowest in the $\aw$ model (Fig. A.1(a)), as expected from the increased awareness level. However, these local dynamics quickly become dominated by the global metastable state, where an increased capacity (higher awareness potential in $\aw$) causes a worse overall performance (compared to the $\io$), similarly to Braess's paradox \cite{braess1968paradoxon} in traffic flows.

\textit{Theoretical results.} The observed awareness paradox and its intuitive explanation can be formalized and proved mathematically, confirming that it is a fundamental phenomenon in behavior-disease models. In this Letter, we focus on presenting the main results and insights, and for further details we refer the reader to Supplementary Sections 3-5.

We obtain theoretical results by making two kinds of approximations on the underlying network and the stochastic process. Firstly, we assume the network is annealed \cite{annealed}, namely, the adjacency matrix is replaced by the connection probabilities $w_{ij}$ modifying the awareness function $a$ to
\begin{align}
\label{eq:a_hat}
\hat{a}(i,t):=\exp \biggl(-\sum_{j}w_{ij} \hat{u}_{j}(t) \biggr),
\end{align}
where $\hat{u}_{j}(t)$ is the indicator of vertex $j$ being infected at time $t$. Analogously to Eq. \eqref{eq:betadef}, a susceptible vertex $i$ gets infected at rate $\beta_0 \hat{\phi}_{i}(t)$ in the annealed model, with  
\begin{align}
\label{eq:phi_hat}
    \hat{\phi}_{i}(t)= \hat{a}(i,t)^{\alpha_{S}} \sum_{j}w_{ij} \hat{a}(j,t)^{\alpha_{I}} \hat{u}_{j}(t).
 \end{align}

Secondly, we apply the N-intertwined mean-field approximation \cite{NIMFA} and assume that the states of the vertices are independent. The resulting closed system of ODEs can be expressed as 
\begin{align}
\label{eq:u}
 \frac{\mathrm{d}}{\mathrm{d} t}u_i(t)=-\mu u_i(t)+\beta_0 \phi_{i}(t) (1-u_i(t)),
 \end{align}
where the symbols without the hat symbol represent the mean field approximation of the expectation of the quantities with the hat symbol. More explicitly, $u_i(t)$ approximates $\mathbb{E}\left(\hat{u}_{i}(t) \right)$, and $\phi_{i}(t)$ is analogous to $\hat{\phi}_i(t)$ with  $\hat{u}_{i}(t)$ replaced by $u_i(t)$ in Eq. \eqref{eq:a_hat} and Eq. \eqref{eq:phi_hat}.

\begin{figure}[b]
\vspace{-1.5em}
\includegraphics[width=0.5\textwidth]{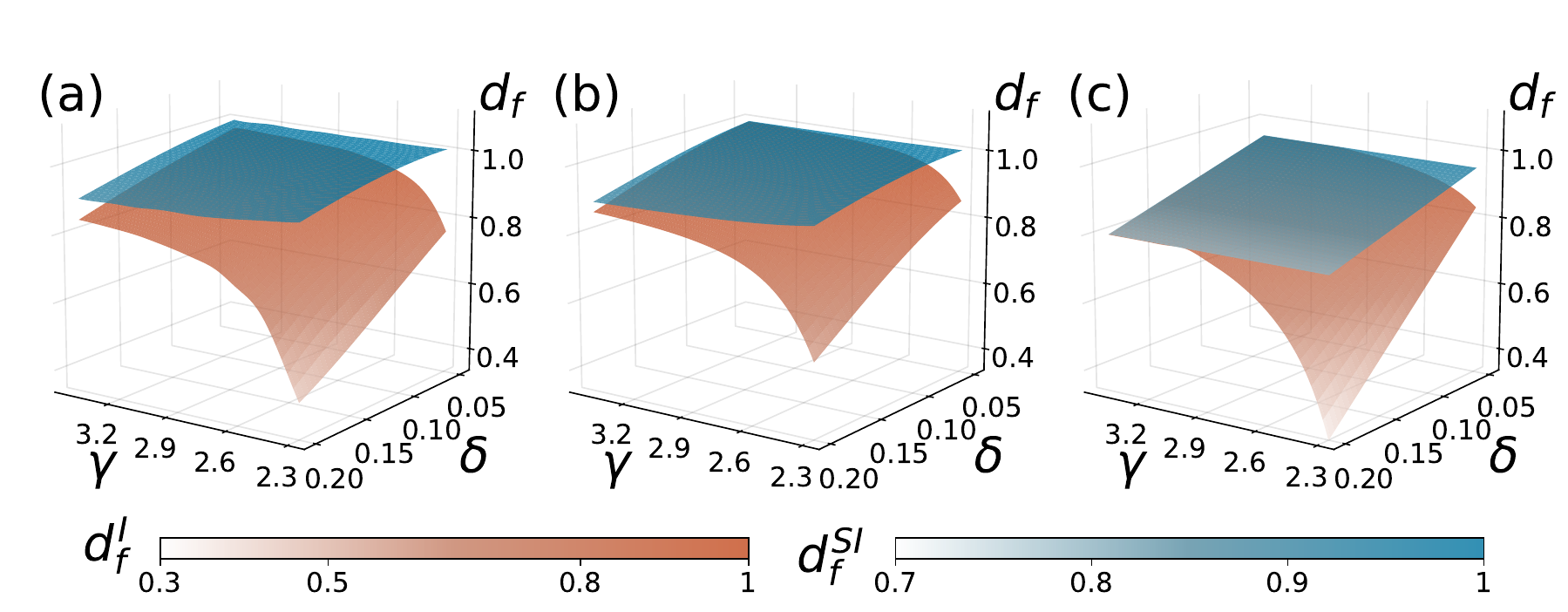} 
\caption{The fractal dimension $d_f$ of the infection size as a function of the degree distribution exponent $\gamma$ and the average degree exponent $\delta$ in the I-aware and the SI-aware models is in agreement in (a) stochastic simulations, (b) mean-field numerical results and (c) analytical results (Eq. \eqref{eq:d_asymp}). For $\gamma\in(2,3)$, the fractal dimension in the I-aware model is strictly smaller than in the SI-aware model in all three approaches. }
\label{fig:3}
\end{figure}

After solving Eq. \eqref{eq:u} (see derivation details at the end of this Letter, in the Appendices), setting degree exponent $\gamma \in (2,3)$ and density parameter $\kappa=n^{\delta}$ for some $0 \leq \delta \leq 1$, we arrive to our main theoretical result:
\begin{align}
\label{eq:d_asymp}
\begin{split}
d_f^\text{SI}=&1-\delta, \\
d_{f}^{\text{I}}=&1-\frac{1}{\gamma-2}\delta. 
\end{split}
\end{align}
Indeed, Eq. \eqref{eq:d_asymp} shows that the fractal dimension of the epidemic is strictly smaller in the $\io$ model compared to the $\aw$ model for heterogeneous ($\gamma \in (2,3)$) and dense ($\delta>0$) degree distributions, proving the paradox.

Fig. \ref{fig:3} confirms the agreement between the stochastic simulations, the numerical approximation and Eq. \eqref{eq:d_asymp}. These quantitative results reinforce our previous observation, that the paradox appears in its strongest form in dense and degree heterogeneous networks. For further illustration, in Supplementary Section 2 we show via simulations and theoretic analysis that the paradox is present in extremely heterogeneous star-like networks~too. 

\textit{Results on real networks.} After demonstrating and understanding our counter-intuitive result on synthetic Chung-Lu networks, we shift our focus towards real networks. We gathered 17 real social and computer networks from publicly available datasets (see Table A.1
). Stochastic simulations revealed that among these 17 networks, 5 exhibit the paradox (flickr follower, libmseti rating, livemocha friendship, marker cafe friendship, github mutual follower). Based on our results in synthetic networks, we expect that the two main network characteristics that distinguish these 5 networks from the other 12 are average degree and degree heterogeneity. Indeed, when we plot the ratio of the infection size in the two models 
\begin{equation}
\label{eq:Irat}
  I_{\text{rat}}=\frac{I_{\text{SI}}}{I_{\text{I}}+I_{\text{SI}}}
\end{equation}
as a function of the average degree and the standard deviation of the degree distribution, there is a clear separation between the 5 paradox-exhibiting and the 12 non-exhibiting networks (see Fig. A.4). However, if we employ a degree-preserving random shuffling of the edges, the paradox appears in 5 additional networks (see Fig. A.4
(b)), suggesting that network characteristics beyond the degree distribution play a role in the appearance of the paradox.

According to our previous results, the key mechanism responsible for the paradox is the infected hubs inducing awareness in low-degree nodes. However, many real networks are assortative, which means that low-degree nodes are less likely to connect to hubs, weakening the paradox. Indeed, plotting the ratio $I_{\text{rat}}$ in Fig. \ref{fig:4} as a function of the average degree $\langle k \rangle$ and the assortativity parameter $\xi$ (defined in the figure legend) we see a clear separation both for real (main) and synthetic networks (upper inset). This result suggests that disassortativty is a key network characteristic that contributes for the paradox to arise, although dense enough assortative  networks still exhibit the phenomenon (Fig. \ref{fig:4} upper inset).

Finding all network characteristics that may contribute to the paradox is outside the scope of this Letter. However, both the mathematical proof and the data analysis suggest that network density, disassortativity and degree heterogeneity are the three most important properties, while other characteristics (e.g. clustering coefficient) play a smaller role (see Table A.2
).

\begin{figure}[h]
\includegraphics[width=1\columnwidth]{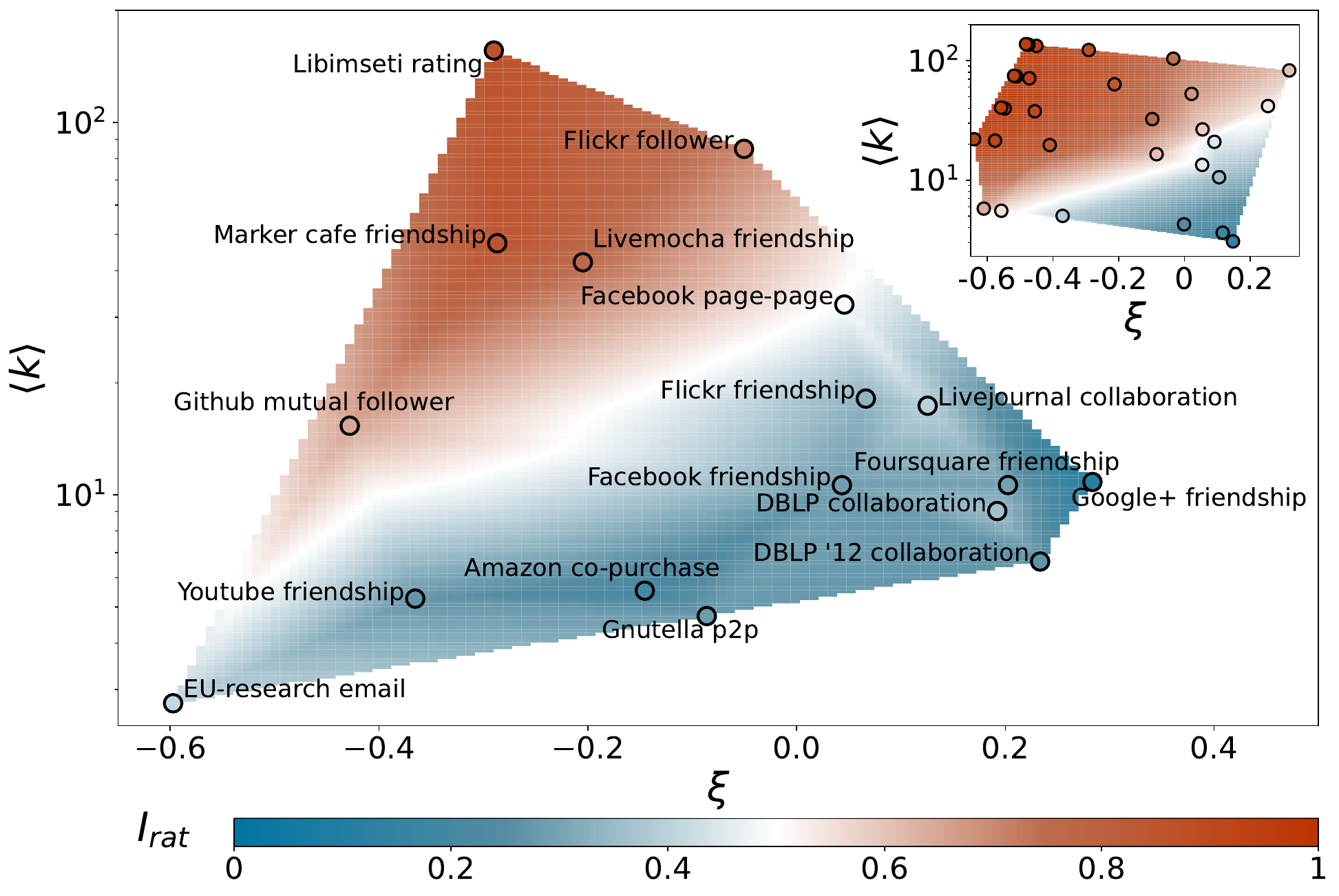}
\caption{The separation between paradox-exhibiting (red) and non-exhibiting (blue) real networks. The main plots shows $I_{\text{rat}}$, defined in Eq. \eqref{eq:Irat}, as a function of the average degree $\langle k \rangle$ and the degree assortativity $\xi$. We compute $\xi$ by fitting the exponent of $k_{\text{nn}}(k) = ck^{\xi}$, where $k_{\text{nn}}(k)$ is the average neighbor degree of nodes with degree $k$ \cite{correlationPastor}. If $\xi > 0$, the network is assortative; otherwise, it is disassortative.
The continuous surface is fitted on the data points via linear interpolation. The upper inset shows a similar behavior on synthetically generated Chung-Lu networks with tunable assortativiy (see in Supplementary Section 8).
}
\label{fig:4}
\end{figure}

\textit{Conclusion.} In this Letter, we demonstrated a highly counter-intuitive result about the interplay of disease propagation and awareness behavior; \textit{fewer} potentially aware nodes may be \textit{more} effective in containing an epidemic, due to a \textit{sublinear} population size dependency of stationary epidemic size, in case only infected nodes are aware.
This is a more effective strategy than \textit{any} other preventive measure, which permits a linear epidemic size (e.g. other types of local awareness, global awareness, government regulations, etc.). The observed paradox reinforces that awareness behavior has the greatest societal benefit when focused on protecting others, as shown for the epidemic threshold earlier \cite{funk2010modelling}. However, instead of thresholds, by focusing on the metastable epidemic size, we provide a very different explanation in nature.
We reveal that the paradox is induced by the interaction between human behavior and the disease distribution within the population, calling for epidemic behavioral surveillance efforts during pandemics beyond reporting the effective reproductive number or sheer case numbers.
 
Although our observations are motivated by surveys, and demonstrated by simulations and mathematical analysis, they have limitations. The implemented awareness models are only a crude approximation of reality. They assume that all individuals employ the same awareness function, while awareness may depend on individual features~\cite{xu2022impact}. A more detailed exploration of the interactions between human behavior and epidemics could be the target of future works. Nevertheless, the fact that the epidemic size is \textit{asymptotically} smaller in the I-aware model compared to the S and SI-aware models, with various real networks exhibiting similar behavior, suggests that the paradox is robust against our modelling choices. 

Further, the degree distribution of real contact networks is often more uniform than the assumed power-law scaling. Moreover, the information about the disease may not be transmitted on the same contact network where the disease spreads. Nevertheless, we believe the paradox will persist in multiplex networks with a more homogeneous contact layer and a scale-free information layer \cite{granell2013dynamical,massaro2014epidemic,kan2017effects}, a topic which remains for future research.

Perhaps our most important message is the delicate role of influential nodes on the long-term dynamics of an epidemic. Asymmetries in the local update rule and the presence of hubs have already been shown to produce surprising results in case of the voter model~\cite{reverse_voter_model}. In epidemic models, Zhang \emph{et al.}~\cite{zhang2014suppression} reported that influential nodes can act as a ``double-edged sword'' by either speeding up or slowing down the disease propagation depending on whether they transmit the disease or the information about the disease at a faster rate. We complement these findings by understanding how different motivations to engage in awareness behavior change the role of influential nodes. We show that if hub nodes are primarily interested in protecting others, they are able to increase the perceived infection density in the population, leading to an awareness response that controls the epidemic more effectively than in scenarios where hubs are primarily motivated by self-protective behaviors. We believe that besides a deeper mathematical understanding of the phenomenon, this finding will inspire future modelling, and it will guide information campaigns in epidemic prevention and crisis management.

\textit{Acknowledgment.} We acknowledge Z. Király, B. Ráth, M. Abért, and L. Lovász for their scientific contributions and thank J. Koltai and G. Röst for their contribution to the MASZK data collection. CS.K. and M.K. were supported by National Laboratory for Health Security (RRF-2.3.1-21-2022-00006); G.Ó. by the Swiss National Science Foundation (P500PT-211129); M.K. by the CHIST-ERA project (SAI: FWF I 5205-N); the SoBigData++ H2020-871042; and the SoBigData-PPP HORIZON-INFRA-2021-DEV-02 program (101079043).

\vspace{20px}

\section*{Appendix} In this section, we provide additional insights for specialists on our theoretical results. 

For sequences $a_n$ and  $b_n$, we use the notations $a_n \sim b_n$ to denote $a_n/b_n=1+o(1)$, and $a_n \asymp b_n$ to denote $a_n=\Theta(b_n)$, (big $\Theta$ is the notation of asymptotical behavior should not be confused with small $\theta$, our notation for the perceived infection density). Whenever we take limits, we first let $n \to \infty$ and then let $\kappa \to \infty.$  Since the paradox only involves the $\aw$ and $\io$ models we assume $\alpha_I>0$ in the analytic derivations.

We are interested in finding the non-zero steady state of Eq. \eqref{eq:u}, which we prove is unique in Lemma A.1., and denote by $u_i$. Interestingly, to solve this system of ODEs, we first have to understand the infection density perceived by the nodes, or in other words, the probability of a randomly selected stub being infected in the mean-field model, defined by
$$\theta= \frac{1}{D}\sum_{i} d_{i} u_i.$$

Solving the system of mean-field equations (detailed derivation in Lemmas A.3, A.4 and Remark A.1),
the asymptotic behavior of $\theta$ can be expressed as

\begin{align}
\label{eq:theta_sim}
\theta \sim \frac{1}{\alpha_{S}+\alpha_{I}} \frac{\log \kappa}{ \kappa}.
\end{align}
This result agrees with the second inset of Fig. \ref{fig:2}(d), and with the intuition that a higher $\alpha_{S}+\alpha_{I}$ should result in a lower infection density at least in the ``perceived'' sense. Next, we show that while the perceived infection density is a good proxy for the epidemic size in certain cases, the epidemic size can be significantly lower in others, inducing the paradox. 

When mainly low degree vertices are infected, the normalized epidemic size $I_{\infty}/D$ and the perceived epidemic density $\theta$ are comparable. In fact, $\theta$ upper bounds $I_{\infty}/D$~as
\begin{equation}
\label{eq:theta_I}
    \theta= \frac{1}{D} \sum_{i} d_i u_i \geq \frac{1}{D} \sum_{i} u_i=\frac{1}{D}I_{\infty}.
\end{equation}
In Eq. \eqref{eq:theta_I}, equality holds if and only if $u_i=0$ for all $i$ larger than one, i.e., if only the lowest degree nodes are infected. As we saw in Fig. \ref{fig:2}(d), this is approximately true in $\aw$, suggesting that for this awareness model, inequality Eq. \eqref{eq:theta_I} is asymptotically sharp. We formalize this intuition (see details in Lemma A.10
) mathematically~as 
 \begin{align}
 \label{eq:I_aa_sim}
 I^{\text{SI}}_\infty \sim D \theta_{\text{SI}} \asymp  n \theta_{\text{SI}} .
 \end{align}

In contrast, in the $\io$ model, the infection probabilities $u_i$ increase for larger degrees, potentially making the inequality Eq. \eqref{eq:theta_I} loose, and therefore $I^{\text{I}}_\infty$ asymptotically smaller than $I^{\text{SI}}_\infty$. Indeed, we show that the tightness of inequality Eq. \eqref{eq:theta_I} undergoes a phase transition in the degree exponent, and the critical threshold is at $\gamma=3$, where the second moment of the degree distribution becomes infinite. More precisely (derivation in Remark A.2
), we can write the asymptotic behavior of the number of infected individuals in model $\io$ as
\begin{align}
\label{eq:I_io_asmyp}
I^{\text{I}}_\infty \asymp \begin{cases}
                        n\theta_{\text{I}}  \ \ &\text{if  $3<\gamma$}, \\
                        n\theta_{\text{I}}^{\frac{1}{\gamma -2}} \ \ &\text{if $2<\gamma<3$} .
                    \end{cases}
 \end{align}

Since $\theta_{\text{I}} \asymp \theta_{\text{SI}} \to 0$ as $\kappa \to \infty$, $I^{\text{I}}_\infty$ is indeed asymptotically smaller than $I^{\text{SI}}_\infty$ in the range $\gamma \in (2,3)$, in particular, for large enough networks we have $I^{\text{I}}_\infty<I^{\text{SI}}_\infty$, proving the paradoxical result. 

Moreover, we are able to quantify the fractal dimension in the $\io$ and $\aw$ models as a function of the degree exponent and the average degree. Setting $\gamma \in (2,3)$ and $\kappa=n^{\delta}$ for some $0 \leq \delta \leq 1$, and combining equations Eq.~\eqref{eq:theta_sim}, Eq.~\eqref{eq:I_aa_sim} and Eq.~\eqref{eq:I_io_asmyp}, we arrive to our main theoretical result expressing the fractal dimensions in Eq.~\eqref{eq:d_asymp}.


\bibliography{apssamp}


\onecolumngrid
\newpage
\appendix 
\counterwithin{figure}{section}
\counterwithin{table}{section}
\counterwithin{lemma}{section}
\counterwithin{theorem}{section}
\counterwithin{remark}{section}

\begin{center}
\noindent \textbf{\LARGE Supplementary Material} \\ \vspace{.2in}

\noindent  \textbf{\Large Epidemic paradox induced by awareness driven network dynamics} \\ \vspace{.2in}

{\large Csegő Balázs Kolok, Gergely \'Odor, Dániel Keliger and M\'arton Karsai\textsuperscript{*} \\ \vspace{.1in}
\small Corresponding author: karsaim@ceu.edu}
\end{center}

\section{}
\label{s:supplementary}

\subsection{Models in details}
In the main script, we use the SIS epidemic model combined with local awareness to model epidemic spreading on networks. In this paper we focus on the awareness function $a(v,t) =  e^{- N_I(v,t)}$, where $N_I(v,t)$ denotes the absolute number of infected neighbors of node $v$ at time $t$ (absolute model). This awareness function has already been studied by Ref. \cite{zhang2014suppression}. In a competing approach, the proportion of infected neighbors is used instead of the absolute number \cite{wu2012impact, LI2020110090, Paarporn, Bagnoli061904}  (proportional model). Our choice to use the absolute model has multiple motivations. First, we aim to capture a scenario where each neighbor may alert node $v$ about their infectious status, and each of these independent alerts cause the same reduction in the infection probability. This criterion is also satisfied by the most common way to model local awareness -- as a social contagion process spreading on the same or a related (information) network as the epidemic process \cite{funk2009spread,perra2011towards,kiss2010impact}. In the social contagion model, the probability staying unaware decreases exponentially in the absolute number of aware neighbors. While we do not adopt the social contagion mechanism in this paper for the sake of mathematical tractability, we chose the absolute model to match this exponential decrease. Second, we observed that the paradox does not appear in the proportional model (plots and calculations not shown), further motivating our focus on the more interesting absolute model. Competing definitions of fractional  \cite{watts2002simple} and absolute  \cite{granovetter1978threshold} threshold models are also provided to describe complex contagion phenomena, leading to different behaviors in their outcome.

Our agent-based model updates every individual's epidemic state during the simulations for every $t$ timestep. We use two methods to simulate the SIS spreading: stochastic simulations and numerical simulations. On the other hand, we have analytical calculations to prove our simulation results. In every case, we have fixed the spreading rate $\beta_0=2 \mu$ with $\beta_0=0.6$ (although this assumption can be easily relaxed by setting a sufficiently large $\kappa$).

\subsubsection{Stochastic simulations}
\begin{description}
    \item[Initialization] \hspace{0em}
    \newline
    \begin{enumerate}
    \vspace{-1.5em}
        \item Fix the network size $n$, power-law exponent $\gamma$ and density parameter $\kappa$. 
        \item Assign weights $d_i=\left(\frac{i}{n}\right)^{-\frac{1}{\gamma-1}}$ to every node $i \in \{1,\dots,n\}$ \cite{fasino2021generating}.
        \item Generate a scale-free network with power-law degree distributions of exponent $\gamma$ via drawing an undirected edge between nodes $i$ and $j$ with probability $\kappa \frac{d_i d_j}{D}$, where $D=\sum_i d_i$.
        \item Initialize the epidemic states choosing independently random $10$ nodes with uniform chances as infected, the rest is susceptible. 
    \end{enumerate}
    \item[Spreading process]\hspace{0em}
    \newline
    \begin{itemize}
    \vspace{-1.5em}
        \item Take timesteps until the standard deviation of the infection density in the last $200$ timestep drops below $0.05$ (this is true at around $1000$ iteration see Fig. \ref{fig:metastate}).
        \item In timestep $t$, update the infectious state of every node. 
        \item If node $u$ is susceptible, then every infected neighbor $v$ infects $u$ by a probability $\beta(u,v,t)=\beta_0 a(u,t)^{\alpha_S} a(v,t)^{\alpha_I}$, which depends on the awareness model we are using ($\so, \io, \aw$) by setting $\alpha_S$ and $\alpha_I$ appropriately. These infections happen independently random from each other, if at least one neighbour infects $u$ then it changes its viral state to infected at the next timestep.
        \item If node $u$ is infected, it recovers and becomes susceptible in the next timestep with probability $\mu$ otherwise it remains infected.
    \end{itemize}
\end{description}

If the network is already provided, the generation phase is unnecessary, and the first three steps of the initialization process can be skipped.
\begin{figure}[htb!]
    \centering
    \includegraphics[width=0.7\textwidth]{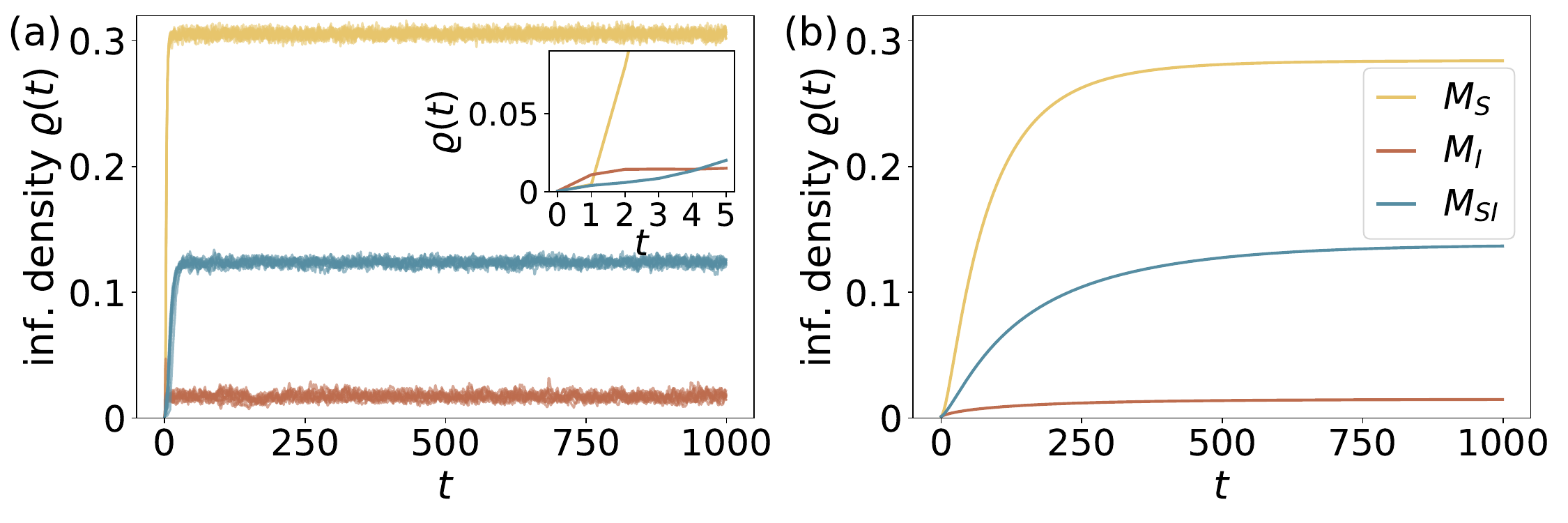}
    \caption{
    Infection density over time for different models: (a) Stochastic simulations show that the infection stabilizes rapidly, presenting the paradox, as the $\io$ model reaches the lowest infection density. The inset provides a closer view of the initial spreading phase, where the $\io$ model starts with the highest density, followed by $\aw$, and $\so$. After a few steps, the order reverses. (b) The numerical solution stabilizes around $t=1000$. All presented populations consist of $10^4$ nodes.}
    \label{fig:metastate}
\end{figure}

\subsubsection{Numerical simulations}
\begin{description}
    \item[Initialization]\hspace{0em}
    \newline
    \begin{enumerate}
    \vspace{-1.5em}
        \item Fix the number of nodes $n$, power-law exponent $\gamma$ and density parameter $\kappa$.
        \item Assign weights $d_i=\left(\frac{i}{n}\right)^{-\frac{1}{\gamma-1}}$ to every node $i \in \{1,\dots,n\}$.
        \item Initialize the $n$ dimensional infectious probability vector $u$ by choosing independently random $10$ nodes with uniform chances as infected (set $u_i=1$ for these), the rest is susceptible (set $u_i=0$ for the rest). 
    \end{enumerate}
    \item[Spreading process] \hspace{0em}
    \newline
    \begin{itemize}
    \vspace{-1.5em}
        \item Take timesteps until the maximum deviation of the infection density from the mean infection density in the last $200$ timestep drops below $0.001$ (this is true at around $1000$ iteration see Fig. \ref{fig:metastate}).
        \item In timestep $t$, update the infectious probability vector assigned to the population by taking a Runge-Kutta step of the differential Eq. (4) described in the main text. 
    \end{itemize}
\end{description}

\subsection{Special case: star-like weighted network}
\label{sec:star-like}

In this section, we demonstrate the role of the hubs in the paradox phenomena via a star-like weighted network. Node $i=1$ plays the role of the hub with weights $w_{11}=0, w_{1i}=1 \ (i>1).$ The rest of the network is homogeneously mixed, that is $w_{ij}=\frac{1}{n-1} \ (i,j>1 \ \textit{including $i=j$ for the sake of simplicity.})$ The ratio of infections within the non-hub nodes are denoted by $\hat{\rho}(t):=\frac{1}{n-1}\sum_{i>1} \hat{u}_i(t) \approx \frac{I(t)}{n}.$ 

In this setup $\hat{\phi}_i(t)$ takes the form of
\begin{align}
\hat{\phi}_1(t)=&(n-1)e^{-\alpha_s (n-1) \hat{\rho}(t)}e^{-\alpha_I(\hat{u}_1(t)+\hat{\rho}(t))} \hat{\rho}(t)), \\
\begin{split}
(i>1) \ \hat{\phi}_{i}(t)=&e^{-\alpha_I(\hat{u}_1(t)+\hat{\rho}(t))} \left(e^{-\alpha_I(n-1)\hat{\rho}(t)}\hat{u}_1(t) \right. \\
&\left.+ e^{-\alpha_I (\hat{u}_1(t)+\hat{\rho}(t))} \hat{\rho}(t) \right) \\
&\approx e^{-(\alpha_I+\alpha_S)(\hat{u}_1(t)+\hat{\rho}(t))} \hat{\rho}(t),
\end{split}
\end{align}
where the approximations work well if the infection levels $\bar{\rho}(t)$ are at least $cn$ for some positive $c$.

When $\alpha_S=0, \alpha_I>0$ then whenever node $i=1$ is susceptible it gets infected at rate $\beta_0 \hat{\phi}_1(t) \asymp n$ which means it gets infected quickly, thus we may assume $\hat{u}_1(t)=1$ for most of $t$ making the effective infection rate for non-hub nodes $i>1$ 
$$\beta_0 \hat{\phi}_i(t) \approx \beta_0 e^{-\alpha_I(1+\hat{\rho}(t))} \hat{\rho}(t)=:\beta_{I}(\hat{\rho}(t)).$$

Meanwhile, if we assume $\alpha_S=\alpha_I>0$ then the susceptible hub $i=1$ gets infected at rate $\beta_0 \hat{\phi}_1(t) \leq \beta_0 (n-1)e^{-\alpha_S(n-1)\hat{\rho}(t)} \to 0$, hence, with large probability it remains susceptible for long time intervals so we may set $\hat{u}_1(t)=0.$ The effective infection rates for the rest of the nodes $i>1$ then becomes
$$\beta_0 \hat{\phi}_i(t) \approx \beta_0 e^{-2\alpha_I\hat{\rho}(t)} \hat{\rho}(t)=:\beta_{SI}(\hat{\rho}(t)).$$

Since for all $0<\hat{\rho}<1$ we have that $\alpha_I(1+\hat{\rho})>2\alpha_I \hat{\rho}$ this means that the effective infection rate are $\beta_{I}(\hat{\rho})<\beta_{SI}(\hat{\rho})$ for such values.

In the limit $n \to \infty$ the infection levels are described by the deterministic dynamics
\begin{align*}
 \frac{\d}{\d t} \rho_{I}(t)=&\beta_I(\rho_I(t))(1-\rho_I(t))\rho_{I}(t)-\mu \rho_{I}(t),   \\
 \frac{\d}{\d t} \rho_{SI}(t)=&\beta_{SI}(\rho_{SI}(t))(1-\rho_{SI}(t))\rho_{SI}(t)-\mu \rho_{SI}(t).
\end{align*}

If we start from equal initial conditions $0<\rho_{SI}(0)=\rho_{I}(0)<\varepsilon$ then we must have $\rho_{I}(t) < \rho_{SI}(t)$ for later times $t > 0$. This is simply because $0<\frac{\d}{\d t} \rho_{I}(0)<\frac{\d }{\d t}\rho_{SI}(0)$ and whenever $\rho_{SI}(t)-\rho_{I}(t)>0$ is small, then the derivatives are $\frac{\d}{\d t} \rho_{I}(t)<\frac{\d }{\d t}\rho_{SI}(t)$, so $\rho_{I}(t)$ can never surpass $\rho_{SI}(t).$

To support our theoretical calculations, we ran stochastic simulations on a star-like network. In the simulations, two hubs are required to maintain the stability of local awareness impact. Additionally, an average degree large enough is needed to enable the epidemic to spread within the non-hub nodes. Thus, to create the network we started with a random regular network of degree 10 and added two nodes as hubs, connecting them to every other node in the network. In this straightforward setting, the paradox can be observed (see Figure \ref{fig:star}). 

\begin{figure}[ht!]
    \centering
    \includegraphics[width=0.35\textwidth]{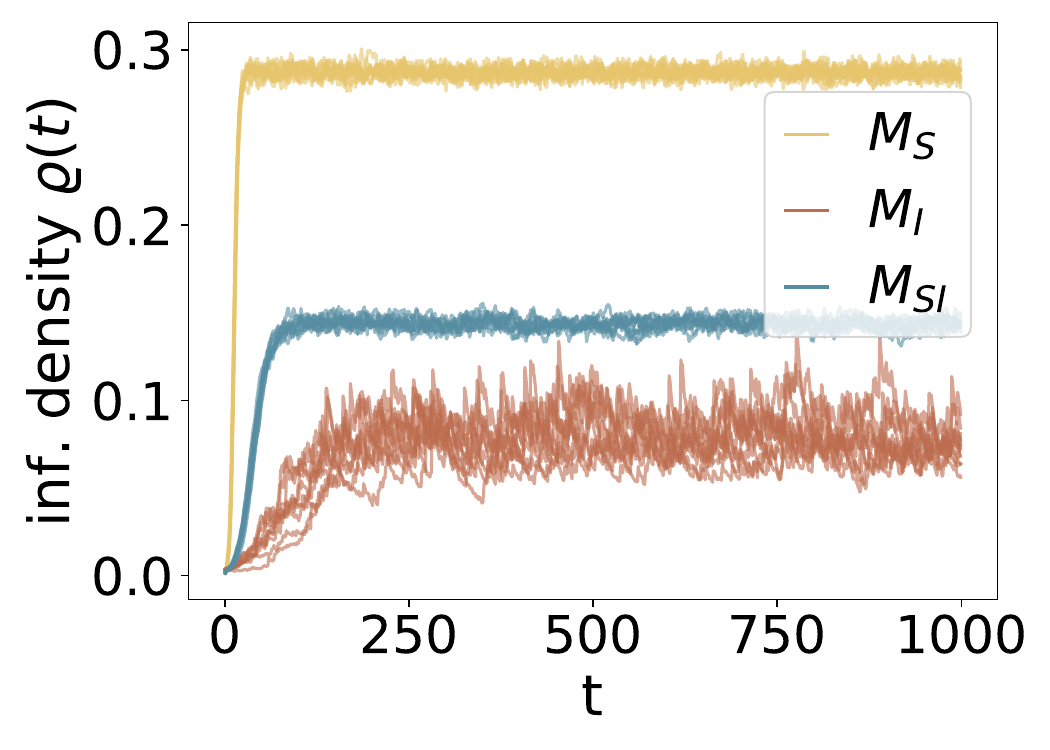}
    \caption{Infection density over time in $10$ runs on a star-like network: a random regular network (degree $10$, size $10^4$) with two hub nodes connected to all others. The $\io$ model paradoxically stabilizes at the lowest infection density in the metastable state.}
    \label{fig:star}
\end{figure}

\subsection{Mean-field equations}

For the sake of compactness we will use the notation $\tau:=\frac{1}{\gamma-1}$ $(0<\tau<1)$ for the exponent of 
$$d_i=\left( \frac{n}{i} \right)^\frac{1}{\gamma-1}=\left( \frac{n}{i} \right)^\tau.$$ Note that the first moment is finite when $0<\tau<1$ and the second moment is infinite for $\frac{1}{2} \leq \tau <1.$

Using Eq. (2) and Eq. (3) the stationary solution $u_i$ must satisfy the following set of equations:

\begin{align}
u_{i}=& \frac{\beta_0 \phi_i}{\mu +\beta_0\phi_i}=:\zeta(\phi_i), \\
\phi_{i}:=& \kappa \eta d_i e^{-\alpha_S \kappa \theta d_i}, \\
\label{eq:theta}
\theta:=&\frac{1}{D} \sum_{i} d_i u_i, \\
\label{eq:eta}
\eta:=&\frac{1}{D} \sum_{i} d_i u_ie^{-\alpha_I \kappa \theta d_i} .
\end{align}
We will think about $u_i, \phi_i$ as a function of $\eta, \theta$: $u_i(\eta,\theta), \phi_{i}(\eta, \theta),$ hence, there are two free parameters that should satisfy the constancy criteria Eq. \eqref{eq:theta}, Eq.  \eqref{eq:eta}.

We start with Eq. \eqref{eq:theta}. Fix and arbitrary $0 \leq \theta \leq 1$. Since $u_i(0,\theta)=0$ and $\eta \mapsto u_i(\eta, \theta)$ is strictly increasing, there must be a unique $\eta=\eta(\theta)$ satisfying
\begin{align}
\label{eq:theta_consistant}
\theta=\frac{1}{D}\sum_{i} d_i u_{i} \left(\eta(\theta),\theta \right).
\end{align}

Now, divide the right hand side of Eq.  \eqref{eq:eta} by $\eta(\theta)$ to define
\begin{align}
\label{eq:G}
G(\theta):= \frac{1}{D}\sum_{i} d_i \frac{u_{i} \left(\eta(\theta),\theta \right)}{\eta(\theta)}e^{-\alpha_I \kappa \theta d_i}.
\end{align}
Clearly, Eq. \eqref{eq:eta} is satisfied if and only if $G(\theta)=1.$

\begin{lemma}
There is a unique $0<\theta<1$ satisfying $G(\theta)=1.$
\end{lemma}

\begin{proof}
 It is easy to check that $\eta(0)=0$ and $\eta(1)=\infty$ which immediately implies $G(0)=\infty$ and $G(1)=0$. Therefore, it is enough to show that $G$ is strictly monotone decreasing.
\begin{align}
\label{eq:phi_derivatives}
\begin{split}
\frac{\partial \phi_i}{\partial \theta}=&-\alpha_S \kappa d_i \phi_i \\
\frac{\partial \phi}{\partial \eta}=&\frac{\phi_i}{\eta} \\
\frac{\d u_i}{\d \theta}=&\zeta'(\phi_i) \phi_i \left( \frac{\eta'}{\eta}-\alpha_S \kappa d_i \right)
\end{split}
\end{align}

Next, we show that $\eta'>0$. Introduce the notation $$\lambda_i:=\frac{1}{D}\zeta'(\phi_i)\phi_i d_i.$$ 
Since $\zeta'>0$ we also get $\lambda_i>0.$ Differentiation both sides of \eqref{eq:theta_consistant} results in
\begin{align*}
1=&\frac{1}{D}\sum_{i}d_i \zeta'(\phi_i)\phi_i \left( \frac{\eta'}{\eta}-\alpha_S \kappa d_i \right), \\
\eta'=&  \frac{1+\alpha_S \kappa \sum_{i}\lambda_i d_i}{\sum_i \lambda_i}\eta>0.
\end{align*}
$G(\theta)$ can be interpreted as the average of the $\frac{u_i}{\eta}e^{-\alpha_S \kappa \theta d_i}$ terms, thus, it is enough to show that each term has a negative derivative.

\begin{align*}
\frac{\d}{\d \theta} \left( \frac{u_i}{\eta}e^{-\alpha_S \kappa \theta d_i} \right)=& \frac{1}{\eta^2} \left[\left(\zeta'(\phi_i)\phi_i \left(\frac{\eta'}{\eta}-\alpha_S \kappa d_i\right)e^{-\alpha_S \kappa \theta d_i}-\alpha_I \kappa d_i u_i e^{-\alpha_S \kappa \theta d_i} \right) \eta-u_ie^{-\alpha_S \kappa \theta d_i} \eta' \right] \\
=& \frac{e^{-\alpha_S \kappa \theta d_i}}{\eta^2} \left[ \left(\zeta'(\phi_i)\phi_i-u_i \right) \eta'-\left(\alpha_S \zeta'(\phi_i)\phi_i+\alpha_I u_i \right)\kappa d_i  \right]
\end{align*}
Note that $\zeta''<0$, hence using the mean value theorem there is a $\xi \in [0, \phi_i]$ such that
$$u_i=\zeta(\phi_i)-\underbrace{\zeta(0)}_{=0}=\zeta'(\xi)\phi_i >\zeta'(\phi_i) \phi_i$$
making all the terms in the bracket negative.
\end{proof}

The main goal is to give bounds on 
\begin{align}
I:= \sum_{i}u_i,
\end{align}
the total amount of infections, and to show that for large enough fixed $\kappa$ we have $I_{\text{I}} <I_{\text{SI}}$ when $n$ is large enough where the subscripts denote the type of awareness model.

Note that strictly speaking, the theorems are not guaranteed to hold for $\kappa=n^{\delta}$ as the $\lim_{\kappa \to \infty} \lim_{n \to \infty}$ limit does not specify how fast $\kappa=\kappa(n)$ could grow. Nevertheless, we carry on with the calculations noting that the formulas may break down for larger $\kappa$ values.

\subsection{ Calculating $\theta$}

The goal of this subsection is to prove Eq. (7) of the main text.
We start with a simple observation.

\begin{lemma}
\label{l:xi}
Assume $\alpha_S>0$. Then there is a $\bar{\phi} \geq 0$ (independent of $\kappa$ and $n$)  such that $ 0 \leq \phi_{i} \leq \bar{\phi}.$

Consequently, $u_i \asymp \phi_i$.
\end{lemma}

\begin{proof}
Clearly, $0 \leq \eta \leq \theta$. Therefore, by setting $0 \leq x:=\alpha_S \kappa \theta d_i$
\begin{align*}
0 \leq \phi_i \leq \kappa \theta d_i e^{-\alpha_S \kappa \theta d_i}=\frac{1}{\alpha_S} xe^{-x} \leq \frac{e^{-1}}{\alpha_I}:=\bar{\phi}.
\end{align*}

For the second half:
\begin{align*}
  \frac{\beta_0}{\mu +\beta_0 \bar{\phi}} \phi_i\leq \frac{\beta_0 \phi_i}{\mu+\beta_0 \phi_i} \leq  \frac{\beta_0}{\mu} \phi_i.
\end{align*}
\end{proof}

\begin{lemma}
\label{l:theta_upper_bound}
For all $0<\varepsilon<1$
$$ \theta<(1+\varepsilon) \frac{\log \kappa }{(\alpha_S+\alpha_I) \kappa}$$
when $n$ and $\kappa$ are large enough.
\end{lemma}

\begin{proof}
Indirectly assume $\theta \geq (1+\epsilon) \frac{\log \kappa}{(\alpha_S+\alpha_I) \kappa}.$

Firstly, assume $\alpha_S>0.$

Using Lemma \ref{l:xi} $\frac{u_i(\eta(\theta),\theta)}{\eta(\theta)} \asymp \kappa d_i e^{-\alpha_S \kappa \theta d_i},$ hence,
\begin{align*}
G(\theta) \asymp& \frac{\kappa}{D} \sum_{i}d_i^2e^{-(\alpha_S+\alpha_I)\kappa \theta d_i} \leq \frac{\kappa}{D}\sum_{i}d_i^2 e^{-(1+\varepsilon) \log \kappa \ d_i}= \\
&\frac{1}{(1+\varepsilon) \log \kappa}\frac{\kappa}{D}\sum_{i}d_i \left( (1+\varepsilon) \log \kappa \ d_i \right)e^{-(1+\varepsilon) \log \kappa \ d_i}.
\end{align*}

For $\kappa \geq e$ we have $(1+\varepsilon) \log \kappa \ d_i \geq 1$. Since $xe^{-x}$ is monotone decreasing for $x \geq 1$ we have that
\begin{align*}
 \leq \frac{1}{(1+\varepsilon) \log \kappa}\frac{\kappa}{D}\sum_{i}d_i \left( (1+\varepsilon) \log \kappa  \right)e^{-(1+\varepsilon) \log \kappa }=\kappa e^{-(1+\varepsilon) \log \kappa }=\kappa^{-\varepsilon} \to 0, 
\end{align*}
resulting in a contradiction.

Secondly, when $\alpha_S=0$ 
\begin{align*}
G(\theta)= \frac{\beta_0 \kappa}{D} \sum_{i} d_i^2 \frac{1}{\mu+\beta_0 \kappa \eta(\theta) d_i }e^{-\alpha_I \kappa \theta d_i} \leq \frac{\beta_0}{\mu} \frac{\kappa}{D}\sum_{i}d_i^2 e^{-\alpha_I \kappa \theta d_i} \leq \frac{\beta_0}{\mu} \kappa^{-\varepsilon} \to 0.
\end{align*}
\end{proof}

\begin{lemma}
\label{l:eta_bound}
Assume $\alpha_S=0$, and $G(\theta)=1$, then $\kappa \eta(\theta) <  \frac{1}{2}$ for large enough $\kappa$. (Also, $ \kappa \eta(\theta) \leq 1.)$
\end{lemma}

\begin{proof}
Indirectly assume $\kappa \eta(\theta) \geq \frac{1}{2}$.

\begin{align*}
&\phi_{i}(\eta(\theta),\theta)= \kappa \eta(\theta) d_i \geq \frac{1}{2} \\
&\theta=\frac{1}{D}\sum_{i}d_i \zeta \left(\phi_i(\eta(\theta),\theta)  \right)\geq \frac{1}{D}\sum_{i}d_i \zeta\left( \frac{1}{2}\right)=\zeta\left( \frac{1}{2}\right)>0
\end{align*}

This leads to a contradiction as $\theta \to 0.$

\end{proof}

\begin{lemma}
\label{l:theta_lower_bound}
For all $0<\varepsilon <1$
$$\theta>(1-\varepsilon) \frac{\log \kappa}{(\alpha_S+\alpha_I) \kappa}$$
for large enough $n$ and $\kappa$.
\end{lemma}

\begin{proof}

Indirectly assume $\theta \leq (1-\varepsilon) \frac{\log \kappa}{(\alpha_S+\alpha_I) \kappa}. $

Firstly, take $\alpha_S>0$. Note that $d_i \leq 1+\varepsilon \ \Leftrightarrow i \geq (1+\varepsilon)^{-\frac{1}{\tau}}n.$

\begin{align*}
G(\theta) \asymp& \frac{\kappa}{D} \sum_{i}d_i^2e^{-(\alpha_S+\alpha_I)\kappa \theta d_i} \geq \frac{\kappa}{D} \sum_{i \geq (1+\varepsilon)^{-\frac{1}{\tau}}n }d_i^2e^{-(1-\varepsilon) \log \kappa \ d_i} \geq \\
& \frac{\kappa}{D}\sum_{i \geq (1+\varepsilon)^{-\frac{1}{\tau}}n } e^{-(1-\varepsilon^2) \log \kappa } \asymp \left(1- (1+\varepsilon)^{-\frac{1}{\tau}}\right) \kappa e^{-(1-\varepsilon^2) \log \kappa } \\
&\asymp \kappa^{\varepsilon^2} \to \infty,
\end{align*}
resulting in a contradiction.

Now assume $\alpha_S=0$. Using Lemma \ref{l:eta_bound},
\begin{align*}
G(\theta)=&\frac{ \kappa}{D} \sum_{i}d_i \frac{\beta_0 d_i}{\mu+\beta_0 \kappa \eta(\theta) d_i}e^{-(\alpha_S+\alpha_I)\kappa \theta d_i} \geq  \frac{ \kappa}{D} \sum_{i}d_i \frac{\beta_0 d_i}{\mu+\beta_0  d_i}e^{-(\alpha_S+\alpha_I)\kappa \theta d_i} \\ 
\geq &  \frac{\beta_0 }{\mu+\beta_0  } \frac{ \kappa}{D} \sum_{i}d_i e^{-(\alpha_S+\alpha_I)\kappa \theta d_i} .
\end{align*}

The rest of the argument is the same as in the $\alpha_S>0$ case as both $d_i^2$ and $d_i$ are lower bounded by $1$ when they are not in the exponent.
\end{proof}

\begin{remark}
\label{r:rho}
Putting Lemma \ref{l:theta_upper_bound} and \ref{l:theta_lower_bound} yields
\begin{align*}
\theta \sim \frac{\log \kappa}{(\alpha_S+\alpha_I) \kappa}
\end{align*}
for both the I- and SI-awareness model.
\end{remark}

\subsection{Calculating $I$}

In this section, we are going to prove Eq. (9) and Eq. (10) of the main text.

\begin{lemma}
\label{l:split}

When $\alpha_S=0$
\begin{align*}
u_i \asymp \begin{cases}
                        1 \ \ \text{if $i \leq \left(\kappa \eta(\theta)\right)^{\frac{1}{\tau}}n$}, \\
                        \kappa \eta(\theta) d_i \ \ \text{if $i \geq \left(\kappa \eta(\theta)\right)^{\frac{1}{\tau}}n$}.
                    \end{cases}
\end{align*}
\end{lemma}

\begin{proof}
Note that $ \phi_i=\kappa \eta(\theta) d_i=1 \ \Leftrightarrow \ i=\left( \kappa \eta(\theta) \right)^{\frac{1}{\tau}}n $.
\begin{align*}
&i \leq ( \kappa \eta(\theta))^{\frac{1}{\tau}}n \ \Rightarrow \ \frac{\beta_0}{\mu+\beta_0} \leq u_i \leq 1 \\
&i \geq \left( \kappa \eta(\theta) \right)^{\frac{1}{\tau}}n \ \Rightarrow \ \frac{\beta_0}{\mu+\beta_0} \kappa \eta(\theta) d_i \leq u_i \leq 
 \frac{\beta_0}{\mu}\kappa \eta(\theta) d_i
\end{align*}
where we used $\kappa \eta(\theta)<1$ according to Lemma \ref{l:eta_bound}.
\end{proof}

\begin{lemma}
\label{l:I_I}

When $\alpha_S=0$
$$I \asymp  \kappa \eta(\theta) n$$
for large enough $n$ and $\kappa$.
\end{lemma}

\begin{proof}
Based on Lemma \ref{l:eta_bound} we know that $\kappa \eta(\theta)<1$ for large enough $n$ and $\kappa$.
\begin{align*}
I=&\sum_{i \leq (\kappa \eta(\theta))^{\frac{1}{\tau}}n} u_i+\sum_{i>(\kappa \eta(\theta))^{\frac{1}{\tau}}n}u_i \asymp (\kappa \eta(\theta))^{\frac{1}{\tau}}n+ \kappa \eta(\theta) \sum_{i>( \kappa \eta(\theta))^{\frac{1}{\tau}}n} d_i \\
=&(\kappa \eta(\theta))^{\frac{1}{\tau}}n+ \kappa \eta(\theta) n^{\tau} \sum_{i>\left(\kappa \eta(\theta)\right)^{\frac{1}{\tau}}n}i^{-\tau} \\ & \asymp ( \kappa \eta(\theta))^{\frac{1}{\tau}}n+\kappa \eta(\theta) n^{\tau} \int_{( \kappa \eta(\theta))^{\frac{1}{\tau}}n}^{n} x^{-\tau} \d x \asymp ( \kappa \eta(\theta))^{\frac{1}{\tau}}n+ \kappa \eta(\theta) n  \asymp \kappa \eta(\theta) n
\end{align*}
\end{proof}

\begin{lemma}

When $\alpha_S=0$

\begin{align*}
\theta \asymp \begin{cases}
                        \kappa \eta(\theta) \ \ \text{if  $0<\tau<\frac{1}{2}$}, \\
                        -\kappa \eta(\theta) \log \kappa \eta(\theta) \ \ \text{if $\tau=\frac{1}{2}$}, \\
                        (\kappa \eta(\theta))^{\frac{1}{\tau}-1} \ \ \textit{if $\frac{1}{2}<\tau<1$} .
                    \end{cases}
\end{align*}
\end{lemma}

\begin{remark}
\label{r:I_io}
 Using Lemma \ref{l:I_I} when $\alpha_S=0$ we have

 \begin{align*}
I_I \asymp \begin{cases}
                        \theta_I n \ \ \text{if  $0<\tau<\frac{1}{2}$}, \\
                        \theta_I^{\frac{\tau}{1-\tau}}n \ \ \textit{if $\frac{1}{2}<\tau<1$} .
                    \end{cases}
 \end{align*}
\end{remark}

\begin{proof}
\begin{align*}
& \theta=\frac{1}{D}\sum_{i}d_i u_i \asymp  \frac{1}{n} \sum_{i \leq (\kappa \eta(\theta))^{\frac{1}{\tau}}n} d_i+\frac{\kappa \eta(\theta) }{n} \sum_{i > (\kappa \eta(\theta))^{\frac{1}{\tau}}n} d_i^2  \\
&\frac{1}{n} \sum_{i \leq (\kappa \eta(\theta))^{\frac{1}{\tau}}n} d_i = n^{-(1-\tau)}\sum_{i \leq (\kappa\eta(\theta))^{\frac{1}{\tau}}n}i^{-\tau} \asymp n^{-(1-\tau)} \left((\kappa \eta(\theta))^{\frac{1}{\tau}}n \right)^{1-\tau}=( \kappa \eta(\theta))^{\frac{1}{\tau}-1} \\
&\frac{\kappa \eta(\theta) }{n} \sum_{i > (\kappa \eta)^{\frac{1}{\tau}}n} d_i^2= \kappa \eta(\theta) n^{-(1-2\tau)} \sum_{i > (\kappa \eta(\theta))^{\frac{1}{\tau}}n}i^{-2 \tau}
\end{align*}

When $\tau=\frac{1}{2}$
\begin{align*}
&\kappa \eta(\theta) n^{-(1-2\tau)} \sum_{i > (\kappa \eta(\theta))^{\frac{1}{\tau}}n}i^{-2 \tau}= \kappa \eta(\theta) \sum_{i>(\kappa \eta(\theta))^2 n}\frac{1}{i} \asymp  \kappa \eta(\theta) \log \frac{n}{(\kappa\eta(\theta))^2 n} \asymp -\kappa \eta(\theta) \log  \kappa \eta (\theta) \\
& \theta \asymp \kappa \eta(\theta)- \kappa \eta(\theta) \log \kappa  \eta(\theta) \asymp -\kappa \eta(\theta) \log \kappa \eta(\theta),
\end{align*}
otherwise,
\begin{align*}
\kappa \eta(\theta) n^{-(1-2\tau)} \sum_{i > (\kappa \eta(\theta))^{\frac{1}{\tau}}n}i^{-2 \tau} \asymp& \frac{1}{1-2 \tau} \kappa \eta(\theta) n^{-(1-2\tau)} \left(n^{1-2\tau}-\left((\kappa \eta(\theta))^{\frac{1}{\tau}}n \right)^{1-2 \tau} \right) \\
=&\frac{1}{1-2 \tau} \left(\kappa \eta(\theta)-(\kappa \eta(\theta))^{\frac{1}{\tau}-1} \right).
\end{align*}

When $0<\tau<\frac{1}{2}$
\begin{align*}
\theta \asymp \kappa \eta(\theta)+(\kappa \eta(\theta))^{\frac{1}{\tau}-1} \asymp \kappa \eta(\theta).
\end{align*}

When $\frac{1}{2}<\tau <1$ the numerator $1-2 \tau$ becomes negative, hence,
\begin{align*}
\theta \asymp (\kappa \eta(\theta))^{\frac{1}{\tau}-1}.
\end{align*}
\end{proof}

\begin{lemma}
\label{l:I_theta_upperbound}
Let us denote the average degree by $\langle d \rangle:= \frac{D}{n} \asymp 1$. Then we have
\begin{align}
\label{eq:I_theta_upperbound}
I \leq \langle d \rangle \theta n.
\end{align}
\end{lemma}

\begin{proof}
\begin{align*}
\theta = \frac{1}{D} \sum_{i} d_i u_i \geq \frac{1}{D} \sum_{i}  u_i= \frac{1}{\langle d \rangle n } I
\end{align*}
\end{proof}

We are going to show that Eq.  \eqref{eq:I_theta_upperbound} is asymptotically exact for the case of the $(M_{\text{SI}})$.

\begin{lemma}
\label{l:I_aa_sim}
Assume $\alpha_S>0$.
\begin{align}
I_{\text{SI}} \sim \langle d \rangle \theta_{\text{SI}} n
\end{align}
\end{lemma}

\begin{proof}
First we show a stronger version of Lemma \ref{l:xi}, that is, we will construct a bound $0 \leq \phi_i \leq \bar{\phi} \to 0$ as $n \to \infty, \ \kappa \to \infty.$ Define 
$$\nu:=\alpha_S \kappa \theta \sim \frac{\alpha_S}{\alpha_S+\alpha_I}\log \kappa .$$

\begin{align*}
\phi_i= \kappa \eta(\theta)d_i e^{-\alpha_S \kappa \theta d_i} \overset{\eta(\theta) \leq \theta}{ \leq} \frac{1}{\alpha_S} \nu d_i e^{-\nu d_i} \leq \frac{1}{\alpha_S} \nu  e^{-\nu}:= \bar{\phi}
\end{align*}
since $\nu d_i \geq \nu \geq 1$ for large enough $\kappa$ and $xe^{-x}$ is monotone decreasing for $x \geq 1$. Furthermore,
\begin{align*}
\underbrace{\frac{\beta_0}{\mu +\beta_0 \bar{\phi}}}_{\to \frac{\beta_0}{\mu}} \phi_i \leq \zeta \left( \phi_i \right) \leq \frac{\beta_0}{\mu} \phi_i
\end{align*}
making $u_i \sim \frac{\beta_0}{\mu} \phi_i$ uniformly in $i$.

\begin{align*}
\frac{D \theta }{I}=\frac{\sum_{i}d_i u_i}{\sum_{i} u_i} \sim \frac{\sum_{i}d_i \phi_i}{\sum_{i} \phi_i}= \frac{\sum_{i}d_i^2 e^{-\nu d_i}}{\sum_{i}d_i e^{-\nu d_i}}=\frac{\frac{1}{n}\sum_{i}d_i^2 e^{-\nu d_i}}{\frac{1}{n}\sum_{i}d_i e^{-\nu d_i}}.
\end{align*}

For $m=1,2$ and fixed $\kappa$ ones has
\begin{align*}
\frac{1}{n}\sum_{i}d_i^m e^{-\nu d_i} \sim \int_{0}^{1}x^{-m\tau}e^{-\nu x^{-\tau}} \d x=:J_m.
\end{align*}

After substituting $y=\nu x^{-\tau}$ we obtain
\begin{align*}
J_m= \frac{1}{\tau} \nu^{\frac{1}{\tau}-m}\int_{\nu}^{\infty}y^{m-\frac{1}{\tau}-1}e^{-y} \d y= \frac{1}{\tau} \nu^{\frac{1}{\tau}-m} \Gamma \left(m-\frac{1}{\tau},\nu \right),
\end{align*}
where $\Gamma(s,x)$ is the incomplete gamma function satisfying the properties
\begin{align*}
\Gamma(s+1,x)=&s \Gamma(s,x)+x^{s}e^{-x}, \\
\lim_{x \to \infty} \frac{\Gamma(s,x)}{x^{s-1}e^{-x}}=1.
\end{align*}

Therefore,
\begin{align*}
\frac{J_2}{J_1}=\frac{1}{\nu} \frac{\Gamma(2-\frac{1}{\tau}, \nu)}{\tau(1-\frac{1}{\tau}, \nu)}=-\left(\frac{1}{\tau}-1\right) \frac{1}{\nu}+\frac{\nu^{-\frac{1}{\tau}e^{-\nu}}}{\tau(1-\frac{1}{\tau}, \nu)} \to 1
\end{align*}
as $\nu \to \infty$ ($\kappa \to \infty).$
\end{proof}

Lastly, we prove the paradox behaviour for large enough $\kappa$ when $\frac{1}{2} \leq \tau <1.$

\begin{theorem}
\label{t:main}
Assume $\frac{1}{2} \leq \tau <1$. Then there is a large enough fixed $\kappa$ such that for all large enough $n$ we have $I_{\text{I}}<I_{\text{SI}}$.
\end{theorem}

\begin{proof}
We know from Remark \ref{r:rho} that $\theta_{\text{SI}}, \theta_{\text{I}} \asymp \frac{\log \kappa}{ \kappa}.$

\begin{align*}
I_{\text{SI}} \asymp& \theta_{\text{SI}} n \asymp \theta_{\text{I}} n \overset{\textit{Lemma } \ref{l:I_I}}{\asymp} \frac{\theta_{\text{I}}}{\kappa \eta(\theta_{\text{I}})} I_{\text{I}} \Rightarrow \\
\frac{I_{\text{SI}}}{I_{\text{I}}} \asymp & \frac{\theta_{\text{I}}}{\kappa \eta(\theta_{\text{I}})}
\end{align*}

When $\tau= \frac{1}{2}$ we have $\theta_{\text{I}} \asymp - \kappa \eta (\theta_{\text{I}}) \log \kappa \eta (\theta_{\text{I}})$. As $\theta_{\text{I}} \to 0$ and according to Lemma \ref{l:eta_bound} $\kappa \eta(\theta_{\text{I}})<\frac{1}{2}$, we must also have $\kappa \eta (\theta_{\text{I}}) \to 0$.

\begin{align*}
\frac{\theta_{\text{I}}}{\kappa \eta(\theta_{\text{I}})} \asymp - \frac{1}{\log \kappa \eta (\theta_{\text{I}})} \to \infty
\end{align*}

When $\frac{1}{2}<\tau<1$ we have $\theta_{\text{I}} \asymp \left(\kappa \eta(\theta_{\text{I}}) \right)^{\frac{1}{\tau}-1},$ which once again implies $\kappa \eta(\theta_{\text{I}}) \to 0$.
\begin{align*}
\frac{\theta_{\text{I}}}{\kappa \eta(\theta_{\text{I}})} \asymp \left(\kappa \eta(\theta_{\text{I}}) \right)^{\frac{1}{\tau}-2} \to \infty
\end{align*}
\end{proof}

\subsection{Nodes with low and high degree over time}

To better understand the paradox, we performed stochastic simulations on a Chung-Lu random network with a degree exponent $\gamma=2.3$, density parameter $\kappa=10$, and size $n=10^4$. We divided the nodes into two groups: the top $1\%$ of nodes with the highest degrees formed the "high" group, while the remaining $99\%$ were classified as the "low" group.

In Figure \ref{fig:supp2}, we plot the infection density over time. Initially, in all three models, the infection spreads mainly among the high-degree nodes. After this initial phase, the effects of awareness differ across the models. In the S-aware and SI-aware models, the infection density in the high group decreases, whereas in the I-aware model, it continues to increase monotonically.

At this point, the infection density in the low group, which makes up $99\%$ of the nodes, reaches a level where it begins to influence the awareness of the high group. In both the $\so$ and $\aw$ models, susceptible high-degree nodes exhibit high levels of awareness because they have several infected low-degree neighbors, making them contract the disease with a very low probability. Conversely, in the $\io$ model, susceptible nodes are not aware by definition, and high-degree susceptible nodes are reinfected by their infected neighbors in the low group with a high chance.

\begin{figure}[hb!]
    \centering
    \includegraphics[width=0.9\linewidth]{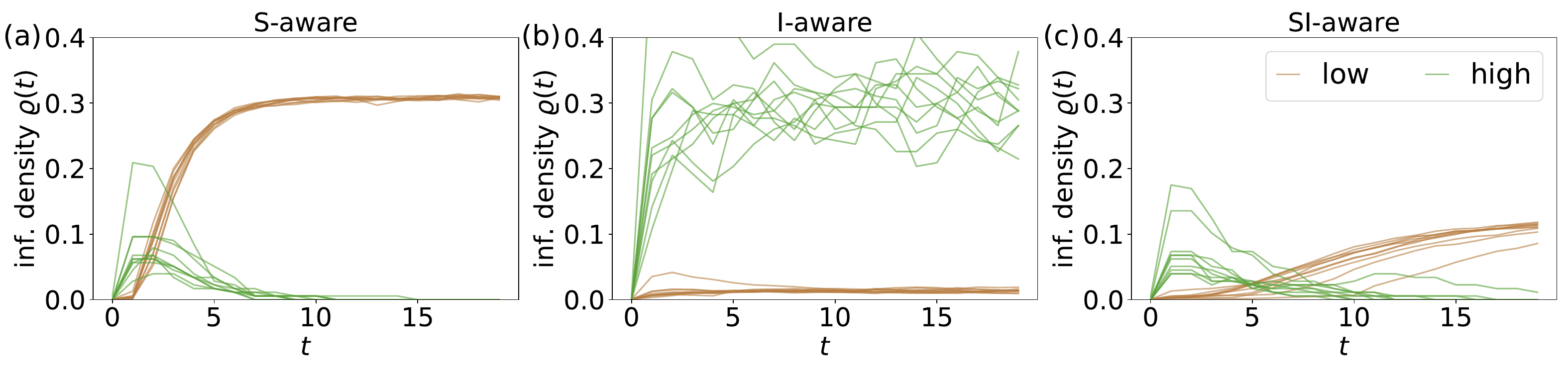}
    \caption{Stochastic simulations on a Chung-Lu network ($\gamma=2.3$, $\kappa=10$, $n=10^4$), with nodes divided into "high" (top $1\%$ with highest degree) and "low" (rest $99\%$) groups. Infection density over time shows initial concentration in the high group, followed by a decline in (a) $\so$ and (c) $\aw$, but continued high density in (b) $\io$.}
    \label{fig:supp2}
\end{figure}

This results in a metastable state where the infection density among high-degree nodes is low in the $\so$ and $\aw$ models but remains high in the $\io$ model. The awareness generated by high-degree nodes significantly influences the number of infections in the low group. In the $\io$ model, infected high-degree nodes raise the awareness of low-degree nodes, reducing infections within the low group. In contrast, the lower awareness induced by high-degree nodes in the $\so$ and $\aw$ models results in more infections within the low group.

\subsection{Results on real networks}
We collected $17$ real networks from publicly available datasets, with their main properties summarized in Table \ref{tab:graph-properties}. Stochastic simulations revealed that $5$ of these networks exhibit the paradox. To identify which network properties might explain the paradox, we plotted the ratio of the $\io$ and $\aw$ metastable epidemic sizes $I_{\text{rat}}$ as a function of the average degree and the standard deviation of the degree distribution (Fig. \ref{fig:stds} (a)). However, the separation between networks that exhibit the paradox and those that do not was not as clear as expected. After applying degree-preserving random shuffling of the edges, the paradox emerged in $5$ additional networks (Fig. \ref{fig:stds} (b)). These previous two observations suggest, that factors beyond heterogeneity and density influence the phenomenon. 

\begin{figure}[h]
    \centering
    \includegraphics[width=0.97\textwidth]{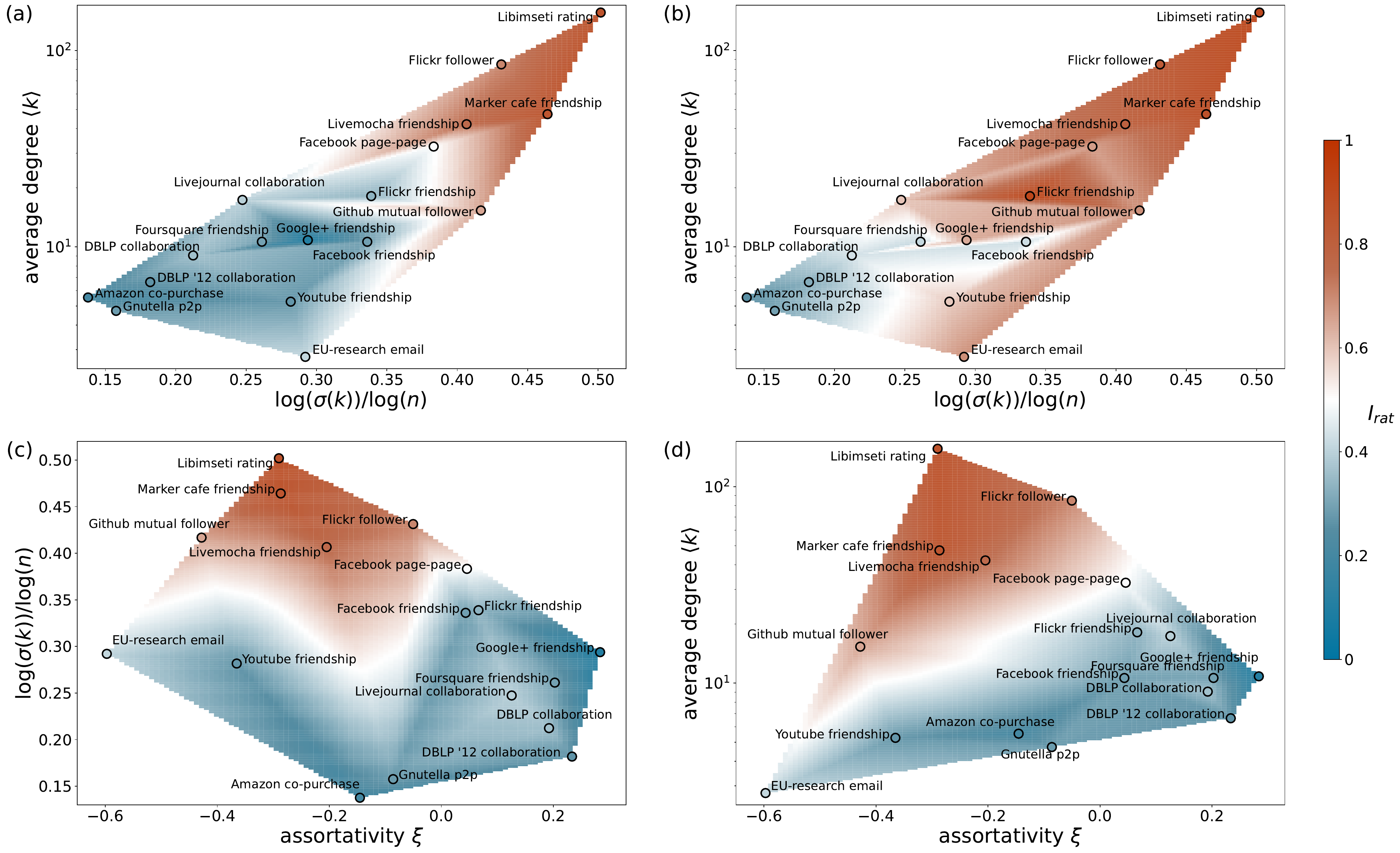}
    \caption{The ratio of metastable epidemic sizes $I_{\text{rat}}=\frac{I_{\text{SI}}}{I_{\text{I}}+I_{\text{SI}}}$ (where $I_{\text{I}}$ and $I_{\text{SI}}$ denote the I-aware and the SI-aware epidemic sizes, respectively). The paradox occurs if the ratio exceeds $0.5$ (red). The continuous surface is fitted on the data points via linear interpolation. (a) The ratio $I_{\text{rat}}$ as the function of the average degree and the standard deviation of the degree distribution of the real networks shows a clear separation. (b) After a degree-preserving random shuffling of the edges, the paradox appears in 5 additional networks. (c) The ratio $I_{\text{rat}}$ as the function of the average degree and the assortativity parameter $\xi$ (defined in the caption of Fig. 4 of the main text) of the real networks shows a clear separation.
    (d) Fig. 4 from the main text.}
    \label{fig:stds}
\end{figure}

From the previous sections, we can conclude that the influence of high-degree nodes on low-degree nodes plays a crucial role in formulating the paradox. However, real-life networks often exhibit assortativity, where nodes with similar degrees are more likely to be connected, weakening the interplay between high- and low-degree nodes. This could explain the absence of the paradox in some cases. Indeed, when we plot the ratio $I_{\text{rat}}$ as a function of the average degree and the assortativity parameter $\xi$, we observe a clearer separation between networks that exhibit the paradox and those that do not.

\setlength{\tabcolsep}{.66em}

\begin{table}[ht!]
    \centering
    \begin{tabularx}{\linewidth}{lrrrrrrrr}
        \toprule
        \textbf{Name}&\textbf{Type} & $\mathbf{n}$ & $\mathbf{\langle k \rangle}$ & \textbf{Med. degree} & \textbf{Max degree/$\mathbf{n}$} & $\mathbf{\xi}$ & $\mathbf{\sigma(k)}$ & $\mathbf{CC}$ \\
        \midrule
        Amazon \cite{amazon} & co-purchase         & 334863    & 5.5  & 4  & 0.00164 & -0.145 & 5.762  & 0.397 \\
        Marker cafe \cite{marker_cafe} &friendship     & 69413     & 47.4 & 6  & 0.12865 & -0.287 & 176.632& 0.186 \\
        Livemocha \cite{livemocha} &friendship       & 104103    & 42.1 & 13 & 0.02863 & -0.205 & 109.675& 0.054 \\
        Gnutella \cite{gnutella}& p2p               & 62586     & 4.7  & 2  & 0.00152 & -0.086 & 5.701  & 0.005 \\
        Flickr \cite{flickrfollower}& follower            & 214626    & 84.9 & 23 & 0.04886 & -0.050 & 199.287& 0.146 \\
        Google+ \cite{marker_cafe}&friendship         & 211187    & 10.8 & 2  & 0.00848 & 0.283  & 36.688 & 0.142 \\
        DBLP '12 \cite{dblp12} &collaboration     & 317080    & 6.6  & 4  & 0.00108 & 0.233  & 10.008 & 0.632 \\
        Facebook \cite{facebookptp}& page-page       & 50517      & 32.4 & 13 & 0.02908 & 0.046  & 63.473 & 0.335 \\
        Facebook \cite{facebook_company}&friendship       & 5793      & 10.6 & 2  & 0.05524 & 0.043  & 18.389 & 0.171 \\
        Github \cite{github}& mutual follower       & 37702      & 15.3 & 6 & 0.25086 & -0.428  & 80.784 & 0.168 \\
        Livejournal \cite{livejournal}&collaboration  & 3997962   & 17.3 & 6  & 0.00371 & 0.126  & 42.957 & 0.284 \\
        Foursquare \cite{foursquare}& friendship      & 114324    & 10.6 & 5  & 0.01063 & 0.203  & 20.946 & 0.179 \\
        DBLP \cite{dblp}&collaboration   & 1836596   & 9.0  & 4  & 0.00121 & 0.192  & 21.381 & 0.631 \\
        EU-research \cite{email_eu}& email          & 265214    & 2.8  & 1  & 0.02879 & -0.597 & 38.369 & 0.067 \\
        Youtube \cite{youtube}&friendship         & 1134890   & 5.3  & 1  & 0.02534 & -0.365 & 50.754 & 0.081 \\
        Libimseti \cite{libimseti}&rating           & 220970    & 156.0& 57 & 0.15110 & -0.290 & 481.272& 0.043 \\
        Flickr \cite{flickr_friendship}& friendship          & 1715255   & 18.1 & 1  & 0.01588 & 0.066  & 129.615& 0.184 \\
        \bottomrule
    \end{tabularx}
    \caption{Network properties of real networks including: network size $n$; average degree $\langle k \rangle$; median degree; normed maximum degree by size; degree correlation exponent $\xi$ \cite{correlationPastor}, derived from fitting $\xi$ in the equation $k_{\text{nn}}(k) = ck^{\xi}$, where $d_{\text{nn}}(k)$ is the average neighbor degree of nodes with degree $k$; standard deviation of the degree sequence $\sigma(k)$; clustering coefficient $CC$.}
    \label{tab:graph-properties}
\end{table}

\begin{table}[h!]
    \centering
    \begin{tabular}{c|c|c|c|c}
         & $\langle k \rangle$ & $\frac{\log(\sigma(k))}{\log(n)}$ & $\xi$ & $cc$   \\
         \hline
         corr. with $I_{\text{rat}}$&0.745 &0.860 & -0.485& -0.332
    \end{tabular}
    \caption{Pearson correlation value between $I_{\text{rat}}$ and average degree $\langle k \rangle$; normalized standard deviation of the degree sequence $\sigma(k)$; degree correlation exponent $\xi$;  clustering coefficient $CC$.}
    \label{tab:corr}
\end{table}

\subsection{Generating Chung-Lu network with tunable assortativity}

In the inset in Fig. 4 of the main text, we required a generative model that allows us to adjust the assortativity of the network to investigate the dependence of the paradox on assortativity. To maintain consistency with the network models described in the main text and \cite{chunglu}, we designed a kernel for the Chung-Lu network that adjusts the assortativity with a single parameter, $\psi$ and at the same time scales the density with $\kappa$. 
The original Chung-Lu network uses node weights $d_i=\left(\frac{i}{n}\right)^{-\frac{1}{1-\gamma}}$ to define the edge probability kernel: $w_{ij}=\kappa \frac{d_i d_j}{D}$, where $D=\sum_i d_i$. 

To create a kernel that scales the assortativity of the network, we define the edge probability $w_{ij}(\psi))$ as: $w_{ij}(\psi)= \kappa \exp(\psi|d_i-d_j|) \frac{d_i d_j}{D}$. If $\psi<0$ the network becomes assortative because nodes with similar degrees are more likely to connect. Conversely, if $\psi>0$ the network becomes disassortative. Using parameter values $\psi=-0.02, -0.0063, -0.0013,  0.0006,  0.0012, 0.0015$ and $\kappa = 1,5,10,20,40$, we generated networks of size $n=10^4$ with assortativity and average degree shown in the inset of Fig. 4 of the main text.

\end{document}